%% file: bicriteria.tex
\documentclass{article}[11pt]
\usepackage{fullpage}
\usepackage{amsmath, amssymb}
\usepackage{amsthm}
\usepackage{graphicx, epsfig}
\usepackage{algorithmic, algorithm}

\newcommand{\OPT}{\operatorname{OPT}}

\newcommand{\rcover}{\mathcal{E}}
\newcommand{\rcoveropt}{\rcover^*}
\newcommand{\assign}{f}

\newcommand{\alg}{\mathcal{A}}
\newcommand{\interval}{\Gamma}

\newcommand{\batch}{B}
\newcommand{\batches}{\mathcal{B}}
\newcommand{\dpath}{R}
\newcommand{\buffer}{\beta}
\newcommand{\yintopt}{\tilde{y}} 
\newcommand{\xintopt}{\tilde{x}} 
\newcommand{\server}{\textsc{Batch-Server}}
\newcommand{\kserver}{$k$-\textsc{Server}}
\newcommand{\reduced}{\textsc{Server}}
\newcommand{\intervals}{\mathcal{I}}

\newtheorem{theorem}{Theorem}

\newtheorem{lemma}[theorem]{Lemma}

\newtheorem{definition}{Definition}
\newtheorem{example}{Example}

\begin{document}

\title{A Bicriteria Approximation for the Reordering Buffer Problem\thanks{This work was supported in part by NSF awards CCF-0643763 and
CNS-0905134.}}


\author{Siddharth Barman\thanks{Computer Sciences Department, University of Wisconsin--Madison. \tt{sid@cs.wisc.edu}} 
\and Shuchi Chawla\thanks{Computer Sciences Department, University of Wisconsin--Madison. \tt{shuchi@cs.wisc.edu}}
\and Seeun Umboh\thanks{Computer Sciences Department, University of Wisconsin--Madison. \tt{seeun@cs.wisc.edu}}
}
\date{}
\maketitle

\begin{abstract}
\input{abstract}
\end{abstract}

\section{Introduction}
\input{introduction}

\section{Notation}
\input{notation}

\input{main}

\input{interval}

\bibliographystyle{plain}
\bibliography{bicriteria}

\appendix
\input{counterexample}

\end{document}

%% file: abstract.tex
In the reordering buffer problem (RBP), a server is asked to process a sequence of requests lying in a metric space. To process a request the server must move to the corresponding point in the metric. The requests can be processed slightly out of order; in particular, the server has a buffer of capacity $k$ which can store up to $k$ requests as it reads in the sequence. The goal is to reorder the requests in such a manner that the buffer constraint is satisfied and the total travel cost of the server is minimized. The RBP arises in many applications that require scheduling with a limited buffer capacity, such as scheduling a disk arm in storage systems, switching colors in paint shops of a car manufacturing plant, and rendering $3$D images in computer graphics.

  We study the offline version of RBP and develop bicriteria approximations. When the underlying metric is a tree, we obtain a solution of cost no more than $9 \OPT$ using a buffer of capacity $4k + 1$ where $\OPT$ is the cost of an optimal solution with buffer capacity $k$. Constant factor approximations were known previously only for the uniform metric (Avigdor-Elgrabli et al., 2012). Via randomized tree embeddings, this implies an $O(\log n)$ approximation to cost and $O(1)$ approximation to buffer size for general metrics. Previously the best known algorithm for arbitrary metrics by Englert et al. (2007) provided an $O(\log^2 k \log n)$ approximation without violating the buffer constraint.

%% file: introduction.tex
We consider the reordering buffer problem (RBP) where a server with buffer capacity $k$ has to process a sequence of requests lying in a metric space. The server is initially stationed at a given vertex and at any point of time it can store at most $k$ requests. In particular, if there are $k$ requests in the buffer then the server must process one of them (that is, visit the corresponding vertex in the metric space) before reading in the next request from the input sequence. The objective is to process the requests in an order that minimizes the total distance travelled by the server. 

RBP provides a unified model for studying scheduling with limited buffer capacity. Such scheduling problems arise in numerous areas including storage systems, computer graphics, job shops, and information retrieval (see~\cite{Racke02onlinescheduling,spieckermann2004sequential,krokowski2004reducing,Blandford2002}). For example, in a secondary storage system the overall performance critically depends on the response time of the underlying disk devices.  Hence disk devices need to schedule their disk arm in a way that minimizes the \emph{mean seek time}. Specifically, these devices receive read/write requests which are located on different cylinders and they must move the disk arm to the proper cylinder in order to serve a request. The device can buffer a limited number of requests and must deploy a scheduling policy to minimize the overall service time. Note that we can model this {disk arm scheduling problem} as a RBP instance by representing the disk arm as a server and the array of cylinders as a metric space over read/write requests. 

The RBP can be seen to be \textrm{NP}-Hard via a reduction from the traveling salesperson problem. We study approximation algorithms. RBP has been considered in both online and offline contexts. In the online setting the entire input sequence is not known beforehand and the requests arrive one after the other. This setting was considered by Englert et al.~\cite{englert2007reordering}, who developed an $O(\log^2 k \log n)$-competitive algorithm. To the best of our knowledge this is the best known approximation guarantee for RBP over arbitrary metrics both in the online and offline case. 

RBP remains \textrm{NP}-Hard even when restricted to the uniform metric (see \cite{corr/abs-1009-4355}). In fact the uniform metric is an interesting special case as it models scheduling of paint jobs in a car manufacturing plant. In particular, switching paint color is a costly operation; hence, paint shops temporarily store cars and process them out of order to minimize color switches. Over the uniform metric, RBP is somewhat related to paging. However, unlike for the latter, simple greedy strategies like First in First Out and Least Recently Used yield poor competitive ratios (see \cite{Racke02onlinescheduling}). Even the offline version of the uniform metric case does not seem to admit simple approximation algorithms. The best known approximation for this setting, due to Avigdor-Elgrabli et al.~\cite{corr/abs-1202-4504}, relies on intricate rounding of a linear programming relaxation in order to get a constant-factor approximation.

The hardness of the RBP appears to stem primarily from the strict buffer constraint; it is therefore natural to relax this constraint and consider bicriteria approximations. We say that an algorithm achieves an $(\alpha,\beta)$ bicriteria approximation if, given any RBP instance, it generates a solution of cost no more than $\alpha \OPT$ using a buffer of capacity $\beta k$.  Here $\OPT$ is the cost of an optimal solution with buffer capacity $k$. There are few bicriteria results known for the RBP. For the offline version of the uniform metric case, a bicriteria approximation of $\left( O(\frac{1}{\epsilon}), 2 + \epsilon \right)$ for every $\epsilon > 0$ was given by Chan et al.~\cite{corr/abs-1009-4355}.  For the online version of this restricted case, Englert et al.~\cite{englert2005reordering} developed a $(4,4)$-competitive algorithm. They further showed how to convert this bicriteria approximation into a true approximation with a logarithmic ratio. We show in Appendix~\ref{appendix:counterexample} that such a conversion from a bicriteria approximation to a true approximation is not possible at small loss in more general metrics, e.g. the evenly-spaced line metric. In more general metrics, relaxing the buffer constraint therefore gives us significant extra power in approximation.


We study bicriteria approximation for the offline version of RBP. When the underlying metric is a weighted tree we obtain a $\left(9, 4 + \frac{1}{k} \right)$ bicriteria approximation algorithm. Using tree embeddings of \cite{FRT} this implies a $\left(O(\log n), 4 + \frac{1}{k} \right)$ bicriteria approximation for arbitrary metrics over $n$ points. 




\paragraph{Other Related Work: } 
Besides the work of Englert et al.~\cite{englert2007reordering}, existing results address RBP over very specific metrics. RBP was first considered by R{\"a}cke et al.~\cite{Racke02onlinescheduling}. They focused on the uniform metric with online arrival of requests and developed an $O(\log^2 k)$-competitive algorithm. This was subsequently improved on by a number of results~\cite{englert2005reordering,avigdor2010improved,adamaszek2011almost}, leading to an $O(\sqrt{\log k})$-competitive algorithm~\cite{adamaszek2011almost}. 

With the disk arm scheduling problem in mind, Khandekar et al.~\cite{khandekar2010online} considered the online version of RBP over the evenly-spaced line metric (line graph with unit edge lengths) and gave an online algorithm with a competitive ratio of $O(\log^2 n)$. This was improved on by Gamzu et al.~\cite{gamzu2007improved} to an $O(\log n)$-competitive algorithm.

Bicriteria approximations have been studied previously in the context of resource augmentation (see~\cite{pruhs2004handbook} and references therein).  In this paradigm, the algorithm is augmented with extra resources (usually faster processors) and the benchmark is an optimal solution without augmentation. This approach has been applied to, for example, paging~\cite{sleator1985amortized}, scheduling~\cite{kalyanasundaram2000speed,bansal2003server}, and routing problems~\cite{Roughgarden02howbad}. 

 \paragraph{Techniques:} 
We can assume without loss of generality that the server is lazy and services each request when it absolutely must---to create space in the buffer for a newly received request. Then after reading in the first $k$ requests, the server must serve exactly one request for each new one received.
Intuitively, adding extra space in the buffer lets us defer serving decisions. 
In particular, while the optimal server must serve a request at every step, we serve requests in batches at regular intervals. Partitioning requests into batches appears to be more tractable than determining the exact order in which requests appear in an optimal solution. This enables us to go beyond previous approaches (see \cite{avigdor2010improved,corr/abs-1202-4504}) that try to extract the order in which requests appear in an optimal solution. 
We enforce the buffer capacity constraint by placing lower bounds on the cardinalities of the batches. In particular, by ensuring that each batch is large enough, we make sure that the server ``carries forward'' few requests. Then the maximum buffer utilization can be bounded by the number of requests carried forward plus the number read before the next batch is processed.

A crucial observation that underlies our algorithm is that when the underlying metric is a tree, we can find vertices $\{ v_i\}_i$ that any solution with buffer capacity $k$ must visit in order. This allows us to anchor the $i$th batch at $v_i$ and equate the serving cost of a batch to the cost of the subtree spanning the batch and rooted at $v_i$. Overall, when the underlying metric is a tree, the problem of finding low-cost batches with cardinality constraints reduces to finding low-cost subtrees which are rooted at $v_i$s, cover all the requests, and satisfy the same cardinality constraints. We formulate a linear programming relaxation, LP$1$, for this covering problem. 


Rounding LP$1$ directly is difficult because of the cardinality constraints. To handle this we round LP$1$ partially to formulate another relaxation that is free of the cardinality constraints and is amenable to rounding. Specifically, using a fractional optimal solution to LP$1$, we determine for each request $j$ an \emph{interval} of indices, $\Gamma(j)$, such that any solution that assigns every request to a batch within its corresponding interval approximately satisfies the buffer constraint. This allows us to remove the cardinality constraints and instead formulate an interval-assignment relaxation LP$2$. In order to get the desired bicriteria approximation we show two things: first, the optimal cost achieved by LP$2$ is within a constant factor of the optimal cost for the given RBP instance; second, an integral feasible solution of LP$2$ can be transformed into a RBP solution using a bounded amount of extra buffer space. Finally we develop a rounding algorithm for LP$2$ which achieves an approximation ratio of $2$.

%% file: notation.tex
An instance of RBP is specified by a metric space over a vertex set
$V$, a sequence of $n$ vertices (requests), an integer $k$, and a
starting vertex $v_0$. The metric space is represented by a graph $G =
(V, E)$ with distance function $d : E \to \mathbf{R}^+$ on edges. We
index requests by $j$. We assume without loss of generality that
requests are distinct vertices. Starting at $v_0$, the server reads
requests from the input sequence into its buffer and clears requests
from its buffer by visiting them in the graph (we say these requests
are served). The goal is to serve all requests, having at most $k$
buffered requests at any point in time, with minimum traveling
distance. We denote the optimal solution as $\OPT$. For the most part
of this paper, we focus on the special case where $G$ is a tree.

We break up the timeline into windows as follows. Without loss of
generality, $n$ is a multiple of $2k+1$, i.e. $n = (2k+1)m$. For $i
\in [m]$, we define window $W_i$ to be the set of requests from $(2k+1)(i-1)
+ 1$ to $(2k+1)i$. Let $w(j)$ be the index of the window in which $j$
belongs. The $i$-th \emph{time window} is defined to be the duration
in which the server read $W_i$.


%% file: main.tex
\section{Reduction to Request Cover Problem}
In this section we show how to use extra buffer space to convert the RBP into a new and simpler problem that we call Request Cover. The key tool for the reduction is the following lemma which states that we can find for every window a vertex in the graph $G$ that must be visited by any feasible solution within the same window. We call these
vertices {\em terminals}. This allows us to break up the server's path into
segments that start and end at terminals.
\begin{lemma}
  \label{lem:terminals}
  For each $i$, there exists a vertex $v_i$ such that all feasible
  solutions with buffer capacity $k$ must visit $v_i$ in the $i$-th time window.
\end{lemma}

\begin{proof}
  Fix a feasible solution and $i$. We orient the tree as follows. For
  each edge $e = (u, v)$, if after removing $e$ from the tree, the
  component containing $u$ 
  contains at most $k$ requests of $W_i$, then we direct the edge from
  $u$ to $v$. Since $|W_i| = 2k + 1$, there is exactly one directed
  copy of each edge.
  
  An oriented tree is acyclic so there exists a vertex $v_i$ with
  incoming edges only. We claim that the server must visit $v_i$
  during the $i$-th time window. During the $i$-th time window, the
  server reads all $2k + 1$ requests of $W_i$. Since each component of
  the induced subgraph $G[V \setminus \{v_i\}]$ contains at most $k$
  requests of $W_i$ and the server has a buffer of size $k$, it cannot
  remain in a single component for the entire time window. Therefore,
  the server must visit at least two components, passing by $v_i$, at
  some point during the $i$-th time window. 
\end{proof}

For the remainder of the argument, we will fix the terminals $v_1,
\ldots, v_m$. Note that since $G$ is a tree, there is a unique path
visiting the terminals in sequence, and every solution must contain
this path.  For each $i$, let $P_i$ denote the path from $v_{i-1}$
to $v_i$.

We can now formally define request covers.

\begin{definition}[Request cover]
  Let $\batches$ be a partition of the requests into \emph{batches}
  $\batch_1, \ldots, \batch_m$, and $\rcover$ be an ordered collection
  of $m$ edge subsets $E_1, \ldots, E_m \subseteq E$. The pair
  $(\batches, \rcover)$ is a \emph{request cover} if
  \begin{enumerate}
  \item For every request $j$, the index of the batch containing $j$
    is at least $w(j)$, i.e. the window in which $j$ is released.
  \item For all $i \in [m]$, $E_i \cup P_i$ is a connected subgraph spanning $B_i$.
  \item There exists a constant $\buffer$ such that for all $i \in
    [m]$, we have $\sum_{l \leq i} |\batch_l| \geq (2k+1)i - \buffer
    k$; we say that the request cover is \emph{$\buffer$-feasible}. We
    call the request cover \emph{feasible} if $\buffer=1$.
  \end{enumerate}
  The {\em length} of a request cover is $d(\rcover) = \sum_i d(E_i)$.
\end{definition}
%

\begin{definition}[Request Cover Problem (RCP)]
  In the RCP we are given a metric space $G=(V,E)$ with lengths $d(e)$
  on edges, a sequence of $n$ requests, buffer capacity constraint
  $k$, and a sequence of $m=n/(2k+1)$ terminals $v_1, \ldots,
  v_m$. Our goal is to find a feasible request cover of minimum length.
\end{definition}

We will now relate the request cover problem to the RBP. Let
$(\batches^*, \rcoveropt)$ denote the optimal solution to the RCP. We
show on the one hand (Lemma~\ref{lem:rcoveropt}) that this solution
has cost within a constant factor of OPT, the optimal solution to
RBP. On the other hand, we show (Lemma~\ref{lem:construct}) that any
$\buffer$-feasible solution to RCP can be converted into a solution to
the RBP that is feasible for a buffer of size $(2+\buffer)k+1$ with a
constant factor loss in length.

\begin{lemma}
  \label{lem:rcoveropt}
  $d(\OPT) \geq d(\rcoveropt)$.
\end{lemma}

\begin{proof}
  For each $i$, let $E_i$ be the edges traversed by the optimal server
  during the $i$-th time window and let $\rcover$ be the collection of
  edge subsets. 
  We have $d(\OPT) \geq \sum_i d(E_i) = d(\rcover)$, so it suffices to
  show that $\rcover = (E_1, \ldots, E_m)$ is a feasible request
  cover. 
  By Lemma \ref{lem:terminals}, both $E_i$ and $P_i$ are connected
  subgraphs containing $v_i$ for each $i$. Hence $\rcover$ is
  connected. Since $E_l$ contains the requests served in the $l$-th
  time window for each $l$, and for each $i$ the server has read $(2k
  + 1)i$ requests and served all except at most $k$ of them by the end
  of the $i$-th time window, we get that $\sum_{l \leq i} |\batch_l| \geq
  (2k + 1)i - k$. This proves that $\rcover$ is a feasible request
  cover. 
\end{proof}

Next, consider a request cover $(\batches, \rcover)$. We may assume
without loss of generality that for all $i$, $E_i \cap P_i =
\emptyset$. This observation implies that $E_i$ can be partitioned
into components $E_i(p)$ for each vertex $p \in P_i$, where $E_i(p)$
is the component of $E_i$ containing $p$.



We will now define a server for the RBP, $\server(\batches, \rcover)$,
based on the solution $(\batches, \rcover)$. Recall that the server
has to start at $v_0$. In the $i$-th iteration, it first buffers all
requests in window $W_i$. Then it moves from $v_{i-1}$ to $v_i$ and
serves requests of $\batch_i$ as it passes by them.

\begin{algorithm}
\caption{$\server(\batches, \rcover)$}
  \begin{algorithmic}[1]
   \label{alg:server}
   \STATE Start at $v_0$
   \FOR {$i = 1$ \TO $m$}
   \STATE (Buffering phase) Read $W_i$ into buffer
   \STATE (Serving phase) Move from $v_{i-1}$ to $v_i$ along $P_i,$
   and for each vertex $p \in P_i$, perform an Eulerian tour of 
   $E_i(p)$. Serve requests of $\batch_i$ along the way. 
   \ENDFOR
  \end{algorithmic}
\end{algorithm}

\begin{lemma}
  \label{lem:construct}
  Given a $\buffer$-feasible request cover $(\batches, \rcover)$,
  $\server(\batches,\rcover)$ is a feasible solution to the RBP
  instance with a buffer of size $(2+\buffer)k + 1$, and has length at
  most $d(\OPT) + 2d(\rcover)$.
\end{lemma}

\begin{proof}
  We analyze the length first. In iteration $i$, the server uses each
  edge of $P_i$ exactly once. Since $E_i$ is a disjoint union of
  $E_i(p)$ for $p \in P_i$, the server uses each edge of $E_i$ twice
  during the Eulerian tours of $E_i$'s components. The total length is therefore
  \[\sum_i d(P_i) + \sum_i 2d(E_i) \leq d(\OPT) + 2d(\rcover).\]

  Next, we show that the server has at most $(2+\buffer)k+1$ requests
  in its buffer at any point in time. We claim that all of $\batch_i$
  is served by the end of the $i$-th iteration. Consider a request $j$
  that belongs to a batch $\batch_i$. Since $i$ is at least as large
  as $w(j)$, the request has already been received by the $i$th
  phase. The server visits $j$'s location during the $i$th iteration
  and therefore services the request at that time if not earlier. This
  proves the claim.

  The claim implies that the server begins the $(i+1)$-th iteration
  having read $(2k+1)i$ requests and served $\sum_{l \leq i}
  |\batch_l| \geq (2k+1)i - \buffer k$ requests, that is, with at most
  $\buffer k$ requests in its buffer. It adds $2k+1$ requests to be
  the buffer during this iteration. So it uses at most $(2+\buffer)k +
  1$ buffer space at all times. 
\end{proof}

\section{Approximating the Request Cover Problem}
We will now show how to approximate the request cover problem. Our
approach is to start with an LP relaxation of the problem, and use the
optimal fractional solution to the LP to further define a simpler
covering problem which we then approximate in Section~\ref{sec:interval}.

\subsection{The request cover LP and the interval cover problem}


The integer linear program formulation of RCP is as follows. To obtain
an LP relaxation we relax the last two constraints to $x(i,j),
y(e,i)\in [0,1]$.
\begin{equation}
\label{lp:rcp}
\tag{LP1}
\boxed{
\begin{aligned}
  \mbox{minimize}\quad            
  & \sum_i \sum_e y(e,i)d_e\\
  \mbox{subject to}\quad  & \sum_{w(j) \leq i} x(j, i) \geq 1 &\quad  \forall j\\
                      & \sum_{j: w(j) \leq i} \sum_{i' \leq i} x(j, i') \geq (2k+1)i - k &\quad \forall i\\
                      & y(e,i) \geq x(j,i) &\quad \forall i, j, e \in
                      R_{ji}\\
                      & x(j, i) \in \{0,1\} &\quad \forall i, j\\
                     & y(e, i) \in \{0,1\} &\quad \forall i,e
\end{aligned} 
}
\end{equation}
Here the variable $x(j,i)$ indicates whether request $j$ is assigned
to batch $\batch_i$ and the variable $y(e,i)$ indicates whether edge
$e$ is in $E_i$. Recall that the edge set $E_i$ along with path $P_i$
should span $B_i$. Let $R_{ji}$ denote the (unique) path in $G$ from
$j$ to $P_i$. The third inequality above captures the constraint that
if $j$ is assigned to $B_i$ and $e\in R_{ji}$, then $e$ must belong to
$E_i$.



Let $(x^*, y^*)$ be the fractional optimal solution to the linear
relaxation of \eqref{lp:rcp}.  
Instead of rounding $(x^*, y^*)$ directly to get a feasible request
cover, we will show that it is sufficient to find request covers that
``mimic'' the fractional assignment $x^*$ but do not necessarily
satisfy the cardinality constraints on the batches (i.e. the second
set of inequalities in the LP). To this end we define an {\em interval
request cover} below.
\begin{definition}[Interval request cover]
  For each request $j$, we define the \emph{service deadline} $h(j) =
  \min\{i \geq w(j) : \sum_{l \leq i} x^*(j,l) \geq 1/2\}$ and the
  \emph{service interval} $\interval(j) = [w(j), h(j)]$. A request
  cover $(\batches, \rcover)$ is an \emph{interval request cover}
  if it assigns every request to a batch within its service intervals.
\end{definition}


In other words, while $x^*$ ``half-assigns'' each request no later
than its service deadline, an interval request cover mimics $x^*$ by
integrally assigning each request no later than its service
deadline. The following is a linear programming formulation for the
problem of finding minimum length interval request covers.
\begin{equation}
 \label{lp:interval}
 \tag{LP2}
\boxed{
\begin{aligned}
  \mbox{minimize}\quad            
  & \sum_i \sum_e y(e,i)d_e\\
  \mbox{subject to}\quad  & \sum_{i \in \interval(j)} x(j, i) \geq 1 &\quad  \forall j\\
                      & y(e,i) \geq x(j,i) &\quad \forall j, i \in \interval(j), e \in R_{ji}\\
                      & x(j, i), y(e,i) \in [0,1] &\quad \forall i, j,
                      e
\end{aligned}
}
\end{equation}

Let $(\xintopt, \yintopt)$ be the fractional optimal of
\eqref{lp:interval}. We now show that interval request covers are
$2$-feasible request covers and that $d(\yintopt) \leq 2d(y^*)$. Since
$d(y^*) \leq d(\rcoveropt)$, it would then suffice to round
\eqref{lp:interval}. 



\begin{lemma}
  \label{lem:2-strict}
  Interval request covers are $2$-feasible.
\end{lemma}

\begin{proof}
  Fix $i$. Let $H_i := \{j : h(j) \leq i\}$ denote the set of all
  requests whose service intervals end at or before the $i$th time
  window. We first claim that $|H_i| \geq (2k+1)i - 2k$. In
  particular, the second constraint of \eqref{lp:rcp} and the
  definition of $H_i$ gives us
  \begin{align*}
    (2k+1)i - k & = \sum_{j : w(j) \leq i} \sum_{i' \leq i } x^*(j, i') 
    = \sum_{j \in H_i : w(j) \leq i}\sum_{i' \leq i} x^*(j, i') + \sum_{j \notin
      H_i : w(j) \leq i} \sum_{i' \leq i} x^*(j, i')\\
    &\leq \sum_{j \in H_i : w(j) \leq i} 1 + \sum_{j \notin H_i : w(j) \leq i} \frac{1}{2}
    = |H_i| + \frac{1}{2}\left((2k+1)i - |H_i|\right).
  \end{align*}
  The claim now follows from rearranging the above
  inequality.

  Note that in an interval request cover, each request in $H_i$ is
  assigned to some batch $B_l$ with $l \leq i$. Therefore,
  \[\sum_{l \leq i} |\batch_l| \geq |H_i| \geq
  (2k+1)i - 2k.\] 
  
\end{proof}

We observe that multiplying all the coordinates of $x^*$ and $y^*$ by
$2$ gives us a feasible solution to \eqref{lp:interval}. Thus we have
the following lemma.
\begin{lemma}
  \label{lem:fract-interval-LP}
  We have $d(\yintopt) \leq 2d(y^*)$.
\end{lemma}

Note that the lemma says nothing about the integral optimal of
\eqref{lp:interval} so a solution that merely approximates the
optimal integral interval request cover may not give a good
approximation to the RBP, and we need to bound the integrality gap of
the LP.
In the following subsection, we show that we can find an interval
request cover of length at most $2d(\yintopt)$.




%% file: interval.tex
\subsection{Approximating the Interval Assignment LP}
\label{sec:interval}
Before we describe the general approximation , we consider two special
cases for insight.

\paragraph{Example: single edge.}
Suppose the tree consists of a single unit-length edge $e = (u,v)$,
all requests reside at $u$, and all terminals at $v$. In this case, $R_{ji} = \{e\}$ for all pairs $j$ and $i$ so
the second set of constraints in \eqref{lp:interval} is simply
\[y(i) \geq x(j,i) \quad \forall j, i \in \interval(j)\] where we
write $y(i)$ for $y(e,i)$. A minimum solution satisfies these
constraints with equality. Summing over $i \in \interval(j)$, we get
that in this case \eqref{lp:interval} is equivalent to
\begin{equation*}
\boxed{
\begin{aligned}
  \mbox{minimize}\quad            
  & \sum_i y(i)\\
  \mbox{subject to}\quad  & \sum_{i \in \interval(j)} y(i) \geq 1 &\quad  \forall j
\end{aligned}
}
\end{equation*}

This is exactly the linear relaxation for the hitting set\footnote{A
  subset $X$ of a universe $U$ is a \emph{hitting set} for
  $\mathcal{S} \subset 2^U$ if $X \cap S \neq \emptyset$ for all $S
  \in \mathcal{S}$.} problem where the sets we want to hit are
intervals. While the general hitting set problem is hard, it turns out
that this special case can be solved exactly in polynomial time and
the relaxation has no integrality gap\footnote{One way to see this is
  that the columns of the constraint matrix has consecutive ones, and
  thus the constraint matrix is totally unimodular.}. Thus, we get an
optimal solution via a reduction to the minimum interval hitting set
problem: compute a minimum hitting set $M$ for the set of intervals
$\intervals := \{\interval(j)\}$, and then add $e$ to $E_i$ for $i \in
M$.

\paragraph{Example: two edges.}
Suppose the tree is a line graph consisting of three vertices $u_1$,
$u_2$ and $v$ with unit-length edges $e_1 = (u_1, v)$ and $e_2 = (u_2,
u_1)$ (Figure~\ref{fig:examples}(a)). Requests reside at $u_1$ and $u_2$, and all terminals at
$v$. 
For each $i$ and $j$ residing at $u_1$, we have $R_{ji} =
\{e_1\}$. For each $i$ and $j$ residing at $u_2$ we have $R_{ji} =
\{e_1, e_2\}$. 
Thus feasible solutions to
\eqref{lp:interval} satisfy the constraints
\begin{align*}
  \sum_{i \in \interval(j)} y(e_1, i) \geq 1 &\quad  \forall j,\\
  \sum_{i \in \interval(j)} y(e_2, i) \geq 1 &\quad \forall j \in u_2.
\end{align*}
The constraints suggest that the vector $y(e_1, \cdot)$ is a
fractional hitting set for the collection of intervals
$\intervals(e_1) := \{\interval(j)\}$, and $y(e_2, \cdot)$ for
$\intervals(e_2) := \{\interval(j) : j \in u_2\}$.
In light of the single-edge special case, a naive approach is to first
compute minimum hitting sets $M(e_1)$ and $M(e_2)$ for
$\intervals(e_1)$ and $\intervals(e_2)$, respectively. Then we add
$e_1$ to $E_i$ for $i \in M(e_1)$, and $e_2$ to $E_i$ for $i \in
M(e_2)$. However, the resulting edge sets may not be connected.
Instead, we make use of the following crucial facts:
\begin{enumerate}
\item[(1)] We should include $e_2$ in $E_i$ only if $e_1 \in E_i$, and,
\item[(2)] Minimal hitting sets are at most twice minimum fractional
  hitting sets (see Lemma \ref{lem:disjoint-intervals}).
\end{enumerate}
These facts suggest that we should first compute a minimal hitting set
$M(e_1)$ for $\intervals(e_1)$ and then compute a minimal hitting set
$M(e_2)$ for $\intervals(e_2)$ with the constraint that $M(e_2)
\subseteq M(e_1)$. This is a valid solution to \eqref{lp:interval}
since $\intervals(e_2) \subseteq \intervals(e_1)$. We proceed as usual
to compute $\rcover$. The resulting $\rcover$ is connected by (1) and
$d(\rcover) \leq 2d(\yintopt)$ by (2).

\begin{figure}
\begin{center}
\begin{tabular}{ccc}
\includegraphics[height=2.0 in]{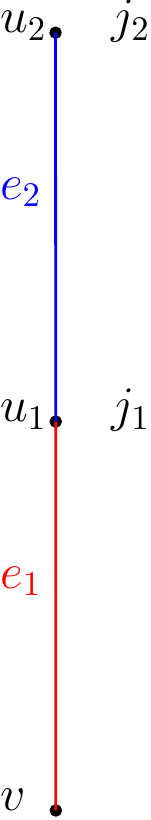}
& \hspace{1 in}  & 
\includegraphics[height=2.0 in]{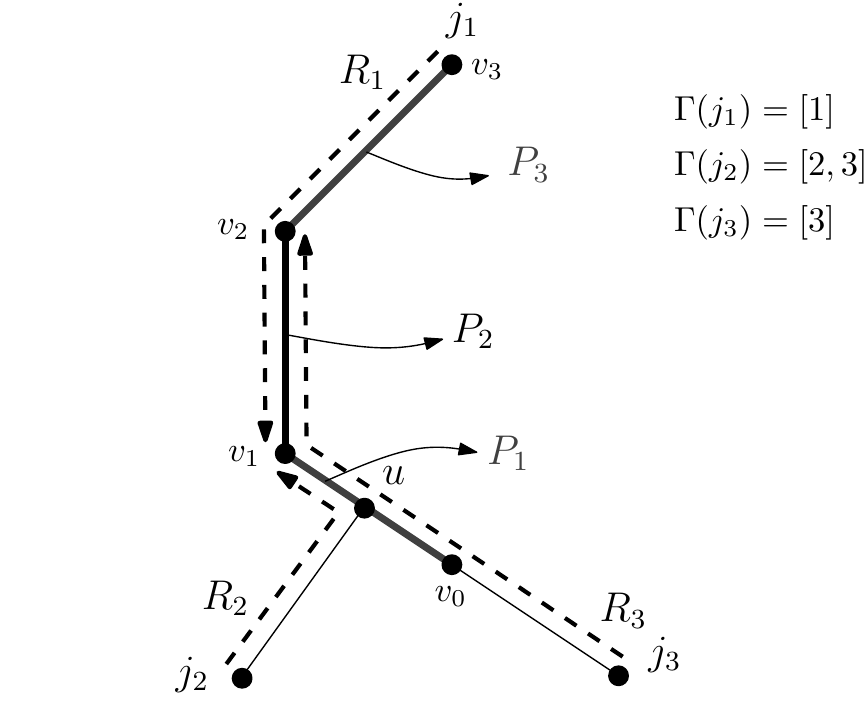}\\
(a) & \hspace{1 in}  &   (b)\\
\end{tabular}
\end{center}
\caption{(a) The two-edge example; (b) $R_i$ is the path from $j_i$ to $P_i$, for $i \in \{1,2,3 \}$. Note that arc $(v_1, v_2)$ precedes $(u,v_1)$ and no arc precedes $(v_1,v_2)$.}   
\label{fig:examples}
\end{figure}

\paragraph{General case.} 

Motivated by the two-edge example, at a high level, our approach for
the general case is as follows:
\begin{enumerate}
\item We construct interval hitting set instances over each edge.
\item We solve these instances starting from the edges nearest to the
  paths $P_i$ first.
\item We iteratively ``extend'' solutions for the instances nearer the
  paths to get minimal hitting sets for the instances further from
  the paths.
\end{enumerate}
We then use Lemma \ref{lem:disjoint-intervals} to argue a $2$-approximation
on an edge-by-edge basis.

Figure~\ref{fig:examples}(b) gives an example of an instance of
interval request cover. Note that whether an edge is closer to some
$P_i$ along a path $R_{ji}$ for some $j$ depends on which direction we
are considering the edge in. We therefore modify \eqref{lp:interval}
to include directionality of edges, replacing each edge $e$ with
bidirected arcs and directing the paths $R_{ji}$ from $j$ to $P_i$.


\begin{equation}
 \label{lp:d-interval}
 \tag{LP2$'$}
\boxed{
\begin{aligned}
  \mbox{minimize}\quad            
  & \sum_i \sum_a y(a,i)d_a\\
  \mbox{subject to}\quad  & \sum_{i \in \interval(j)} x(j, i) \geq 1 &\quad  \forall j\\
                      & y(a,i) \geq x(j,i) &\quad \forall j, i \in
                      \interval(j), a \in R_{ji}\\
                      & x(j, i), y(a,i) \in [0,1] &\quad \forall i, j,
                      a
\end{aligned}
}
\end{equation}
For every edge $e$ and window $i$, there is a single orientation of
edge $e$ that belongs to $R_{ji}$ for some $j$. So there is a 1-1
correspondence between the variables $y(e,i)$ in \eqref{lp:interval}
and the variables $y(a,i)$ in \eqref{lp:d-interval}, and the two LPs are
equivalent. Henceforth we focus on \eqref{lp:d-interval}.


Before presenting our final approximation we need some more notation.

\begin{definition}
  For each request $j$, we define $\dpath_j$ to be the directed path from
  $j$ to $\bigcup_{i \in \interval(j)} P_i$. For each arc $a$, we
  define $C(a) = \{j : a \in \dpath_j\}$
  and the set of intervals $\intervals(a) = \{\interval(j) : j \in
  C(a)\}$
  . We
  say that $a$ is a \emph{cut arc} if $C(a) \neq \emptyset$.

  We say that an arc $a$ \emph{precedes} arc $a'$, written $a \prec
  a'$, if there exists a directed path in the tree containing both the
  arcs and $a$ appears after $a'$ in the path.
\end{definition}

\begin{lemma}
  \label{lem:fractional-feasible}
  Feasible solutions $(x,y)$ of \eqref{lp:d-interval} satisfy the
  following set of constraints for all arcs $a$:
  \begin{equation*}
      \sum_{i \in \interval(j)} y(a,i) \geq 1 \quad \forall j \in
      C(a).
  \end{equation*}
\end{lemma}

\begin{proof}
  Let $(x,y)$ be a feasible solution of \eqref{lp:d-interval}. Fix an
  arc $a$ and 
  $j \in C(a)$. For each $i \in \interval(j)$, we have $a \in R_{ji}$
  since $R_j$ is a path from $j$ to a connected subgraph containing
  $P_i$. By feasibility, we have $y(a,i) \geq x(j,i)$. Summing over
  $\interval(j)$, we get $\sum_{i \in \interval(j)} y(a,i) \geq
  \sum_{i \in \interval(j)} x(j,i) \geq 1$ where the last inequality
  follows from feasibility. 
\end{proof}



We are now ready to describe the algorithm. At a high level, Algorithm
\ref{alg:greedy-ext} does the following: initially, it finds a cut arc
$a$ with no cut arc preceding it and computes a minimal hitting set
$M(a)$ for $\intervals(a)$; iteratively, it finds a cut arc $a$ whose
preceding cut arcs have been processed previously, and minimally
``extends'' the hitting sets $M(a')$ computed previously for the
preceding arcs $a'$ to form a minimal hitting set $M(a)$.


\begin{algorithm}
\caption{Greedy extension}
  \begin{algorithmic}[1]
   \label{alg:greedy-ext}
    \STATE $U \gets \{a : C(a) \neq \emptyset\}$ 
    \STATE $A_i \gets \emptyset$ for all $i$
    \STATE $M(a) \gets \emptyset$ for all arcs $a$
    \WHILE {$U \neq \emptyset$}
    \STATE Let $a$ be any arc in $U$
    \WHILE {there exists $a' \prec a$ in $U$}
    \STATE $a \gets a'$
    \ENDWHILE
    \STATE Let $a = (u,v)$
    \STATE $F(a) \gets \{i : v \in P_i\} \cup
    \bigcup_{w: (v,w) \prec a} M((v,w))$
    \STATE Set $M(a) \subseteq F(a)$ to be a minimal hitting set
    for the intervals $\intervals(a)$
    \STATE $A_i \gets A_i \cup \{a\}$ for all $i \in M(a)$
    \STATE $U \gets U \setminus \{a\}$
    \ENDWHILE
    \STATE $\assign(j) \gets \min\{i \in \interval(j) : \text{$j$
      incident to $A_i$ or $P_i$}\}$ for all $j$
    \STATE $B_i \gets \{j : \assign(j) = i\}$ for all $i$
    \RETURN $\alg = (A_1, \ldots, A_m)$, $\batches = (B_1, \ldots, B_m)$
  \end{algorithmic}
\end{algorithm}

We prove that Algorithm \ref{alg:greedy-ext} actually manages to
process all cut arcs $a$ and that $F(a)$ is a hitting set for
$\intervals(a)$. First, we make the following observation.

\begin{lemma}
  \label{lem:invariants}
  For each iteration, the following holds.
  \begin{enumerate}
  \item If $U \neq \emptyset$, the inner `while' loop finds an arc.
  \item $F(a)$ is
    a hitting set for the intervals $\intervals(a)$.
  \end{enumerate}
\end{lemma}

\begin{proof}
  Since we have a bidirected tree and an arc does not precede its
  reverse arc, 
  the inner `while' loop does not repeat arcs and hence it stops with
  some arc. This proves the first statement.

  We prove the second statement by induction on the algorithm's
  iterations. In the first iteration, the set $U$ consists of cut arcs
  so $a' \nprec a$ for all cut arcs $a'$. Therefore, for all
  $\interval(j) \in \intervals(a)$, $a$ is the arc on $R_j$ closest to
  $\bigcup_{i \in \interval(j)} P_i$ and $v \in \bigcup_{i \in
    \interval(j)} P_i$. This proves the base case. Now we prove the
  inductive case.
  Fix an interval $\interval(j) \in \intervals(a)$. If $a$ is the arc
  on $R_j$ closest to $\bigcup_{i \in \interval(j)} P_i$ and $v \in
  \bigcup_{i \in \interval(j)} P_i$, then $F(a) \cap \interval(j) \neq
  \emptyset$. If not, then there exists a neighboring arc $(v,w) \in
  R_j$ closer to $\bigcup_{i \in \interval(j)} P_i$. We have that
  $\interval(j) \in \intervals((v,w))$ and $(v,w) \prec a$. Since the
  algorithm has processed all cut arcs preceding $a$, by the inductive
  hypothesis we have $F((v,w)) \cap \interval(j) \neq \emptyset$. This
  implies that $M((v,w))$ is a hitting set for $\intervals((v,w))$ and
  so $F(a) \cap \interval(j) \neq \emptyset$.
  Hence, 
  $F(a)$ is a hitting set for $\intervals(a)$.  
\end{proof}

Let $E_i$ be the set of edges whose corresponding arcs are in $A_i$
and $\rcover = (E_1, \ldots, E_m)$, i.e. the undirected version of
$\alg$.
\begin{lemma}
  \label{lem:connected}
  $(\batches, \rcover)$ is an interval request cover.
\end{lemma}

\begin{proof}
  The connectivity of $E_i \cup P_i$ follows from the fact that the
  algorithm starts with $A_i = \emptyset$, and in each iteration an
  arc $a = (u,v)$ is added to $A_i$ only if $v \in P_i$ or $v$ is
  incident to some edge previously added to $A_i$.

  Now it remains to show that $\assign(j) \in \interval(j)$ for all
  requests $j$, i.e. that there exists $i \in \interval(j)$ such that
  $j$ is incident to $A_i$ or $P_i$. If $\dpath_j = \emptyset$, then
  $j \in \bigcup_{i \in \interval(j)} P_i$. On the other hand if
  $\dpath_j \neq \emptyset$, then let $a \in \dpath_j$ be the arc
  incident to $j$. 
  Since the algorithm processes all cut arcs, we have $a \in
  \bigcup_{i \in \interval(j)} A_i$ and thus $j$ is incident to
  $\bigcup_{i \in \interval(j)}A_i$. In both cases, we have
  $\assign(j) \in \interval(j)$.  
\end{proof}

Next, we analyze the cost of the algorithm. Let $D(a)$ be the number
of disjoint intervals in 
$\intervals(a)$.
\begin{lemma}
  \label{lem:disjoint-intervals}
  $D(a) \geq |M(a)|/2$ for all arcs $a$.
\end{lemma}

\begin{proof}
  Let $i_1 < \ldots < i_{|M(a)|}$ be the elements of $M(a)$. For each
  $1 \leq l \leq |M(a)|$, there exists an interval $\interval(j_{l})
  \in \intervals(a)$ such that $M(a) \cap \interval(j_{l}) =
  \{i_{l}\}$, because otherwise $M(a) \setminus \{i\}$ would still be
  a hitting set, contradicting the minimality of $M(a)$. We observe
  that the intervals $\interval(j_{l})$ and $\interval(j_{l+2})$ are
  disjoint since $\interval(j_l)$ contains $i_l$ and
  $\interval(j_{l+2})$ contains $i_{l+2}$ but neither contains
  $i_{l+1}$. 
  Therefore, the set of $\lceil |M(a)|/2 \rceil$ intervals
  $\{\interval(j_{l}) : \text{$1 \leq l \leq |M(a)|$ and $l$ odd}\}$
  is disjoint.  
\end{proof}




\begin{lemma}
  $d(\rcover) \leq 2 d(\yintopt)$.
\end{lemma}

\begin{proof}
  Fix an arc $a$. From Lemmas \ref{lem:fractional-feasible} and
  \ref{lem:disjoint-intervals}, we get
  \[\sum_i \yintopt(a,i) \geq D(a) \geq |M(a)|/2.\]
  Since $d(\rcover) = d(\alg)$, we have
  \begin{align*}
    d(\rcover) 
    &= \sum_a |M(a)| \cdot d_a\\
    &\leq \sum_a \left(2\sum_i \yintopt(a,i)\right) \cdot d_a 
    = 2d(\yintopt).
  \end{align*} 
\end{proof}



Together with Lemmas \ref{lem:2-strict} and
\ref{lem:fract-interval-LP}, we have that $(\batches, \rcover)$ is a
$2$-strict request cover of length at most $4d(y^*) \leq
4d(\rcoveropt)$. Lemmas \ref{lem:rcoveropt} and \ref{lem:construct}
imply that $\server(\batches, \rcover)$ travels at most $9\OPT$ and
uses a buffer of capacity $4k+1$. This gives us the following theorem.

\begin{theorem}
  \label{thm:trees}
  There exists an offline $\left(9, 4+\frac{1}{k}\right)$-bicriteria
  approximation for RBP when the underlying metric is a weighted tree.
\end{theorem}
Using tree embeddings of \cite{FRT}, we get
\begin{theorem}
  \label{thm:general}
  There exists an offline $\left(O(\log n),
    4+\frac{1}{k}\right)$-bicriteria approximation for RBP over
  general metrics.
\end{theorem}

%% file: counterexample.tex
\section{Gap Between Bicriteria and True Approximations}
\label{appendix:counterexample}
In this section we prove that there exists an instance on the
evenly-spaced line metric in which the optimal offline solution with a
buffer of size $k/4$ 
has to travel $\Omega(k)$ times the distance of the optimal offline
solution with a buffer of size $k$.

We consider a line graph $L$ with $2^k$ vertices $p_1 < \ldots <
p_{2^k}$ and unit-length edges. The input is a sequence of requests
described by a binary tree $S$ of depth $k$. Let $r$ be the root of
$S$. We denote the subtree rooted at a vertex $v$ by $S(v)$.
Let 
$l_i$ be the $i$-th leaf according to the preordering of the tree.
We define the \emph{destination label} of vertex $v$ to be $t(v) =
\max \{i : l_i \in S(v)\}$ and the \emph{origin label} of $v$ to be
$s(v) = \min \{i : l_i \in S(v)\}$. That is, $t(v)$ and $s(v)$ are the
highest and lowest indices of any leaf in the subtree rooted at $v$,
respectively. 
The input sequence is constructed as follows. First we obtain the
sequence of vertices according to the preordering of the tree.  Then
we replace each non-leaf vertex $v$ in the sequence with a request
lying at $p_{t(v)}$ on the line, and each leaf vertex $l_i$ with a
block (which we refer to as a \emph{leaf block}) of $k$ requests lying
at $p_i$ on the line. 
For non-leaf vertices, we overload notation and use $v$ to refer both
to the vertex in the binary tree and the corresponding request.

Let $\OPT(k)$ and $\OPT(k/4)$ be the optimal offline solutions to the
above input sequence that use buffers of capacity $k$ and $k/4$,
respectively. 
\begin{theorem}
  \label{thm:counterexample}
  We have $\OPT(k/4) \geq \Omega(k) \OPT(k)$.
\end{theorem}


\begin{example}
  For $k=2$, the line metric is represented by the integers $1,2,3,4$
  and the input sequence is $4,2,1,1,2,2,4,3,3,4,4$.
\end{example}

\begin{figure}
\begin{center}
\includegraphics[height=2.0in]{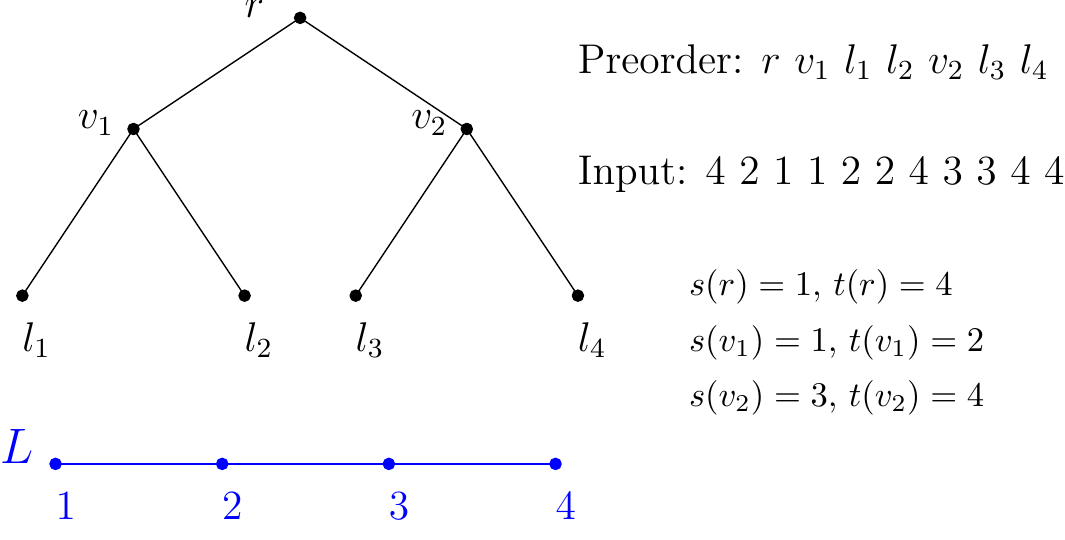}
\caption{Example for $k=2$}
\label{fig:gap-example}
\end{center}
\end{figure}


We present a server $\kserver$ that uses a buffer of size at most $k$
and travels a distance of $2^k-1$ on the above input sequence.
 \begin{algorithm}
  \caption{$\kserver$}
  \begin{algorithmic}[1]
   \label{alg:optk}
   \FOR {$i = 1$ \TO $2^k$}
   \STATE Move to $p_i$
   \STATE Serve all requests $v$ from the input that has $t(v) = i$
   \ENDFOR
  \end{algorithmic}
\end{algorithm}

\begin{lemma}
  On the above input sequence, $\kserver$ uses a buffer of size at
  most $k$ and travels a distance of $2^k - 1$.
\end{lemma}


\begin{proof}
  Since $\kserver$ visits each vertex of the line graph exactly once,
  it travels a distance of $2^k-1$.

  At the beginning of the $i$-th iteration, $\kserver$ has just
  finished reading the $(i-1)$-th leaf block and is at
  $p_i$. Furthermore, it has also served all requests that reside at
  $p_1, \ldots, p_{i-1}$. Thus, it needs to maintain in its buffer
  only the requests $v$ up till the $i$-th leaf block that have $t(v)
  > i$. Since the input sequence is constructed using the preordering
  of $S$, these requests 
  correspond to the ancestors of leaf $l_i$ in $S$. The tree $S$ is of
  depth $k$ so it needs to maintain at most $k-1$ requests in its
  buffer at all times, in addition to a space of $1$ that is needed to
  read requests from the input.
  %
\end{proof}


Next, we show that $\OPT(k/4) \geq \frac{k}{4}\OPT(k)$. Let $\reduced$
be an optimal server with a buffer of size $k/4$ for the above input
sequence. We analyze the movement of $\reduced$ in phases. We define
the $i$-th phase to be the duration starting from the time the last
request of the $(i-1)$th leaf block is read to the time the last
request of the $i$-th leaf block is read.
\begin{lemma}
  \label{lem:window}
  At the end of the $i$-th phase, $\reduced$ is at $p_i$.
\end{lemma}

\begin{proof}
  Since the requests of the $i$-th leaf block all lie at $p_i$ on the
  line metric, we assume w.l.o.g. that either the entire block is
  served together or buffered together. However, the block is of
  length $k$, thus $\reduced$ must serve the entire block.
\end{proof}


For request $v$, let $d(v)$ be the distance travelled between
$p_{t(v)}$ and the previously served request, and $D(v) = \sum_{u \in
  S(v)} d(v)$. Then, the total cost to serve non-leaf vertices is
$D(r) = \sum_v d(v)$. 
We define $C_i$ to be the contents of $\reduced$'s buffer at the end
of the $i$-th phase, i.e. when it reads the last request of the $i$-th
block. Let $C_i(v) = C_i \cap S(v)$ and $C(v) = \sum_{i \in [s(v),
  t(v)]} |C_i(v)|$.


Let $h(v)$ denote the height of $v$.
\begin{lemma}
  \label{lem:induction}
  We have $C(v) + D(v) \geq \frac{h(v)}{2}2^{h(v)}$ for all
  vertices $v$.
\end{lemma}

\begin{proof}
  We use a proof by induction on the height of $v$. For the base case,
  $v$ is a leaf. The base case follows from the fact that a leaf has
  height $0$ and both $C(v)$ and $D(v)$ are non-negative. We consider
  the inductive case next. Let $v_1$ and $v_2$ be the left and right
  children of $v$, respectively. Request $v$ is read in the
  $s(v)$-th phase. Suppose $v$ is served in the $i'$-th phase. 
  Since $C_i(v) =
  C_i(v_1) \cup C_i(v_2) \cup \{v\} $ if $v \in C_i(v)$ and
  $C_i(v) = C_i(v_1) \cup C_i(v_2)$ if $v \notin C_i(v)$, we get
  that
  \begin{align*}
    C(v) &= \sum_{i \in [s(v), t(v)]} |C_i(v_1)|+ |C_i(v_2)| + |[s(v), i'-1]|.
  \end{align*}

  We observe that $[s(v), t(v)] = [s(v_1), t(v_1)] \cup [s(v_2),
  t(v_2)]$, $s(v_1) = s(v)$ and $t(v_2) = t(v)$. Suppose that $i' \in
  [s(v_1), t(v_1)]$ and the request served just before $v$ is
  $v'$. Lemma \ref{lem:window} implies that the server is at
  $p_{i'-1}$ at the beginning of the $i'$-th phase, so w.l.o.g.
  $t(v') \leq t(v)$. The input sequence is obtained using the
  preordering of $S$, so the server has not read any request $u$ with
  $t(u) \in [s(v_2),t(v_2))$. Hence, we have $t(v') < s(v_2)$ so the
  server must have traversed at least $p_{s(v_2)}, p_{s(v)+1}, \ldots,
  p_{t(v_2)}$ to serve $v$. So, we have that $d(v) \geq 2^{h(v)-1}$.
  
  On the other hand, if $i' \in [s(v_2), t(v_2)]$ then, $|[s(v),
  i'-i]| \geq |[s(v_1), t(v_1)]|= 2^{h(v)-1}$. Thus, either
  $|[s(v),i'-1]|$ or $d(v)$ is at least $2^{h(v)-1}$. Since $D(v) =
  D(v_1) + D(v_2) + d(v)$, we get
  \begin{align*}
    C(v) + D(v) 
    &= C(v_1) + C(v_2) + |[s(v), i'-1]| + D(v_1) + D(v_2) + d(v)\\
    &\geq C(v_1) + C(v_2) + D(v_1) + D(v_2) + 2^{h(v)-1}\\
    &\geq \frac{(h(v)-1)}{2}2^{h(v)} + 2^{h(v)-1}\\
    &= \frac{h}{2}2^{h(v)}, 
  \end{align*}
where the second inequality follows from applying the inductive
hypothesis on both $v_1$ and $v_2$. 
\end{proof}

$\reduced$ cannot buffer more than $k/4$ requests at any point in time
therefore $|C_i| \leq k/4$ for all $i$. Applying Lemma
\ref{lem:induction} to the root $r$ implies that $D(r) \geq
\frac{k}{2}2^k - \frac{k}{4}2^k = \frac{k}{4}2^k$. Since 
$\OPT(k/4) \geq D(r)$, this completes the proof of Theorem
\ref{thm:counterexample}.